\documentclass[journal,onecolumn]{IEEEtran}

\usepackage{amsfonts}
\usepackage{color}

\usepackage{amsmath}
\usepackage{amssymb}
\usepackage{graphicx}
\usepackage{blindtext}
\usepackage{titlesec}
\usepackage{enumitem}
\usepackage{authblk}
\usepackage{hyperref}%
\newcommand{\FundingLogos}{%
  \IfFileExists{eu_emblem.png}{\raisebox{0pt}{\includegraphics[height=1.5cm]{eu_emblem.png}}}{}%
  \hspace{1em}%
  \IfFileExists{erc_logo.png}{\raisebox{0pt}{\includegraphics[height=1.5cm]{erc_logo.png}}}{}%
}

\usepackage{tcolorbox}
\tcbset{colback=white,boxrule=0.5pt,arc=2pt,left=4pt,right=4pt,top=4pt,bottom=4pt}
\providecommand{\U}[1]{\protect\rule{.1in}{.1in}}
\newtheorem{theorem}{\textbf{Theorem}}

\newtheorem{corollary}[theorem]{Corollary}

\newtheorem{definition}[theorem]{\textbf{Definition}}
\newtheorem{example}[theorem]{Example}

\newtheorem{lemma}[theorem]{\textbf{Lemma}}

\newtheorem{proposition}[theorem]{\textbf{Proposition}}
\newtheorem{remark}[theorem]{\textbf{Remark}}

\newenvironment{proof}[1][Proof]{\noindent\textbf{#1.} }{\ \rule{0.5em}{0.5em}}

\newif\ifrevised
\newcommand{\revised}[1]{%
	\ifrevised
	\color{blue} #1 \color{black} %
	\else
	#1%
	\fi}
\revisedfalse

\newcommand{\norm}[1]{\left\| #1 \right\|}
\newcommand{\normsq}[1]{\left\| #1 \right\|^2}
\newcommand{\inner}[2]{\left< #1 , #2 \right>}

\let\originalleft\left
\let\originalright\right
\renewcommand{\left}{\mathopen{}\mathclose\bgroup\originalleft}
\renewcommand{\right}{\aftergroup\egroup\originalright}

\ifCLASSINFOpdf

\else

\fi

\begin{document}

\title{\revised{Concavity of Tsallis Entropy and Tsallis Entropy Power along Heat Flow}}

\author{Lukang~Sun
\thanks{L. Sun was with the Department of Mathematics, Technical University of Munich, Munich,
Germany e-mail: lukang.sun@tum.de.}
}


\maketitle

\begin{abstract}
We study the evolution of Tsallis entropy along the heat flow and establish concavity results in arbitrary dimensions. Extending earlier one-dimensional results, we prove that Tsallis entropy is concave along the heat flow for $q\in(0,3]$ in dimension one and for $q\in[1,3]$ in higher dimensions. The upper endpoint $q=3$ is sharp in every dimension. The proof is based on a nonlinear transformation of the heat equation, a sharp dimension-free functional inequality with constant $C_u=3$, and a rigorous justification of the integration-by-parts identities used in the argument. The sharp inequality is proved by an explicit integration-by-parts sum-of-squares identity, rather than by a computer-assisted semidefinite-programming search. As consequences, we recover a generalized de Bruijn identity, prove monotonicity of the associated $q$-Fisher information along the heat flow, and establish concavity results for Tsallis entropy power, including the Shannon entropy-power case and Costa's EPI as an endpoint. We also obtain an asymptotic entropy-power concavity statement for general initial data and a sharp auxiliary functional inequality which may be of independent analytic interest.
\end{abstract}

\IEEEpeerreviewmaketitle

\section{Introduction}

The interplay between entropy and diffusion processes has long been a central theme in information theory, probability, and analysis. A classical example is the evolution of Shannon entropy along the heat flow. This evolution is closely connected with the de Bruijn identity, Fisher information, and the entropy power inequality, and has important applications in communication theory, statistics, and high-dimensional probability; see, for instance, \cite{barron1986entropy,rioul2010information,valero2017generalization}. Beyond first-order monotonicity, the study of higher-order behavior of entropy functionals along diffusion semigroups has also attracted considerable attention. In particular, Costa's concavity theorem for Shannon entropy power along the heat flow \cite{costa1985new} gives a refinement of the entropy power inequality and has motivated many further developments on entropy, diffusion, and functional inequalities.

In parallel, generalized entropy functionals have been studied extensively as non-additive extensions of the Shannon framework; see, for example, \cite{cheng2015higher,hung2022generalization,ledoux2016heat,ledoux2021log,bukal2022concavity,wibisono2018convexity,wu2025completely,savare2014concavity}. Among them, Tsallis entropy is a particularly important example. It depends on an entropic index $q$ and is commonly used to describe non-Gaussian, heavy-tailed, long-range dependent, or non-equilibrium phenomena. From an information-theoretic viewpoint, Tsallis entropy leads naturally to generalized Fisher information and entropy-power functionals. From an analytic viewpoint, it is closely related to diffusion equations and smoothing estimates. Thus the study of Tsallis entropy along the heat flow provides a natural way to test which parts of the classical Shannon theory remain valid in the non-additive setting.

Despite these motivations, the concavity behavior of Tsallis entropy along the heat flow is not fully understood in higher dimensions. Previous works have obtained detailed information on higher-order derivatives of Tsallis and R\'enyi entropies mainly in one dimension; see \cite{hung2022generalization,wu2025completely}. However, the extension to several dimensions is not straightforward. One-dimensional identities often rely on scalar differential structures which do not have direct analogues in higher dimensions. In higher dimensions, Hessian terms, Laplacian terms, and mixed gradient-Hessian terms interact in a more complicated way. The main purpose of this paper is to overcome this difficulty and prove a dimension-uniform second-order concavity theorem for Tsallis entropy along the heat flow.

Our main result shows that Tsallis entropy is concave along the heat flow for $q\in(0,3]$ in dimension one and for $q\in[1,3]$ in all higher dimensions. The upper endpoint $q=3$ is sharp in every dimension. As consequences of the proof, we obtain the generalized de Bruijn identity for Tsallis entropy and prove monotonicity of the corresponding $q$-Fisher information along the heat flow. We also establish concavity results for Tsallis entropy power. In particular, the entropy-power theorem includes the Shannon entropy-power case and recovers Costa's EPI as the endpoint $q=1$ with the classical normalization.
{
\paragraph{The SOS framework and the present method.}
Several works closely related to the present paper use the sum-of-squares
(SOS) framework to study signs of entropy derivatives along the heat flow; see,
for example,
\cite{cheng2015higher,guo2022lower,hung2022generalization,wu2025completely}.
In the usual SOS strategy, one first rewrites the entropy derivative, after
integration by parts, as a quadratic form in a vector of differential
monomials, together with possible lower-order nonnegative terms. The sign
problem is then reduced to proving positive semidefiniteness of an associated matrix under algebraic constraints. This approach is powerful, but the
number of monomials and constraints grows rapidly with the order of the
derivative and with the dimension. Consequently, in many cases the construction
of the SOS representation relies on computer-assisted semidefinite programming,
curve fitting, or other algebraic search procedures.

The proof in the present paper is closely connected with this framework. It
relies on integration by parts together with square completions, and therefore
can naturally be viewed as an argument within the SOS framework. The novelty is not that the final positivity
mechanism lies outside the SOS framework, but rather the way in which the
relevant SOS identity is found. Instead of starting from a large vector of
derivative monomials and searching for a positive semidefinite  matrix, we
introduce the nonlinear variable $u_t=\phi_t^{q/2}$. This transforms the
second-derivative calculation into a simple normal form in which the dependence
on the entropy parameter is separated from the analytic inequalities. The key
functional inequalities are then obtained from a small number of
integration-by-parts identities and elementary square completions.

This gives a new analytic route for constructing the required SOS identity. In
particular, the method produces the sharp dimension-free constant $C_u=3$
directly and extends without any additional algebraic search to arbitrary
dimensions. In this sense, the present approach should be viewed as an analytic
construction of an SOS identity through the nonlinear diffusion variable, rather
than as a departure from the SOS philosophy. Its practical advantage over the
standard SOS strategy is that it avoids the high-dimensional monomial-matrix
search and makes the dimension-uniform structure of the problem transparent.
\paragraph{Necessity and application value of the results.}
The results are useful because heat flow is the analytic model for Gaussian smoothing, or equivalently for adding independent Gaussian noise. Understanding how Tsallis entropy, Tsallis-Fisher information, and Tsallis entropy power evolve under this operation is relevant whenever the classical Shannon entropy is not the most suitable measure of uncertainty. This occurs, for example, in models with heavy tails, non-Gaussian fluctuations, long-range dependence, or non-equilibrium statistics. The generalized de Bruijn identity identifies the rate of increase of Tsallis entropy under Gaussian smoothing with a Tsallis-Fisher information functional. The monotonicity result then shows that this generalized Fisher information dissipates along the heat flow, giving a quantitative smoothing principle in the non-additive setting. The entropy-power concavity result is also natural from an information-theoretic perspective. In the Shannon case, entropy power concavity is closely related to Gaussian extremality and additive Gaussian noise channels. The Tsallis entropy-power result developed here gives an analogue in a non-additive regime and may be useful for studying generalized entropy power inequalities, stability under Gaussian perturbations, and large-noise behavior of non-Gaussian distributions. The asymptotic concavity statement is especially relevant in this respect: even when a smallness condition on the initial density is not imposed, the heat flow eventually enters a regime in which the Tsallis entropy power becomes concave.

\paragraph{Recent counterexamples and interpretation of the second derivative.}
Recent counterexamples to complete-monotonicity-type conjectures show that sign patterns for entropy derivatives along the heat flow must be interpreted carefully. After the first version of this work, Gu and Sellke \cite{GuSellke2026} constructed a one-dimensional counterexample to the Gaussian completely monotone conjecture, showing that a higher-order derivative of Shannon entropy can have the opposite sign from the conjectured one. Zou, Fan, Gao, and Wang \cite{ZouFanGaoWang2026} constructed a two-dimensional counterexample to the log-convexity of Fisher information along the heat flow, with consequences in higher dimensions by tensorization. These developments do not contradict the present results, because our theorem concerns a second-order concavity property of Tsallis entropy and Tsallis entropy power in specified parameter ranges. Rather, they clarify the scope of the result. The second derivative of Tsallis entropy should not be viewed as part of an unrestricted complete-monotonicity principle. Instead, its sign is a finite-order, parameter-dependent, and dimension-dependent property. The sharp endpoint $q=3$ and the distinction between the one-dimensional range and the higher-dimensional range are consistent with this viewpoint. The recent counterexamples therefore strengthen the motivation for proving precise second-order statements with explicit parameter ranges, rather than relying on broad expectations about all higher derivatives.

The works most closely related to ours include \cite{hung2022generalization,wu2025completely,guo2022lower,cheng2015higher}. The papers \cite{hung2022generalization,wu2025completely} study higher-order derivatives of Tsallis and R\'enyi entropies along the heat flow in one dimension. The works \cite{cheng2015higher,guo2022lower} focus on the Shannon entropy setting and use SOS-type methods to study entropy derivatives and related conjectures in low dimensions. Compared with these works, the present paper focuses on a second-order Tsallis entropy problem but treats arbitrary dimensions. The central point is that the nonlinear transformation reveals a simple dimension-uniform structure which is hidden in the original heat-equation variables.

We emphasize that the present work does not attempt to establish a general complete-monotonicity theorem for Tsallis entropy. Higher-order derivatives would require new identities and new functional inequalities, and the recent counterexamples mentioned above indicate that unrestricted sign patterns should not be expected without additional assumptions. The contribution of the paper is instead a sharp second-order concavity theorem for Tsallis entropy, together with entropy-power consequences and functional inequalities that may be useful in further work on generalized entropy functionals.
}
\paragraph{Notation}
For a matrix \(A\), \(\|A\|^{2}=\sum _{i,j}a_{ij}^{2}\) denotes the squared Frobenius norm; for a vector \(x\in\mathbb{R}^d\), \(\|x\|^{2}=\sum _{i=1}^dx_{i}^{2}\) denotes the squared Euclidean norm. Unless otherwise specified, all integrations are performed over the entire Euclidean space \(\mathbb{R}^{d}\).

\paragraph{Organization}
The remainder of this paper is organized as follows. Section \ref{sec:1} introduces the setting, the main definitions, and the statements of the main results. Section \ref{sec:2} justifies the integration-by-parts identities used throughout the paper. Section \ref{sec:3} proves the Tsallis entropy concavity theorem and the sharp functional inequalities. Section \ref{sec:5} proves the concavity results for Tsallis entropy power. Section \ref{sec:4} contains concluding remarks.
\section{Problem Set Up and Main Results}\label{sec:1}

Given a non-negative initial density \(\phi _{0}\) with \(\|\phi _{0}\|_{L^{1}}=1\), define the heat flow \(\phi _{t}:=\phi _{0}*p_{t}\), where \(*\) denotes the convolution operator and \(p_{t}(x):=(4\pi t)^{-\frac{d}{2}}e^{-\frac{\|x\|^{2}}{4t}}\) represents the standard Gaussian heat kernel. It follows that \(\phi _{t}\) is the solution to the heat equation:
\begin{equation}\label{eq:ieee1}\partial_t\phi_t = \Delta\phi_t, \quad \lim_{t \to 0}\phi_t = \phi_0.\end{equation}
It can be verified that for any \(t>0\), the above defined \(\phi _{t}\) is positive with $\norm{\phi_t}_{L^1}=1$ and smooth on \(\mathbb{R}^{d}\).

\begin{definition}\label{def:ieee1}For a probability density function \(\rho \), the Tsallis entropy with entropic index \(q\in (0,1)\cup (1,\infty )\) is defined as:\begin{equation}S_q(\rho) = \frac{1 - \int \rho^q(x) dx}{q-1}.\end{equation}\end{definition}
In the limit as \(q\rightarrow 1\), \(S_{q}(\rho )\) converges to the Shannon entropy (or Boltzmann-Gibbs entropy):\begin{equation}H(\rho) = -\int_{\mathbb{R}^d} \rho(x) \log \rho(x)  dx.\end{equation}

\paragraph{Concavity of Tsallis Entropy}
We now state the main theorems of this paper.
\begin{theorem}\label{thm:ieee1}
Along the heat flow \(\phi_t\), the Tsallis entropy \(S_q(\phi_t)\) is concave in time, namely
\begin{equation}
\frac{d^2}{dt^2}S_q(\phi_t)\le 0,
\end{equation}
for the entropic index \(q\) in the following ranges:
\begin{itemize}\addtolength{\itemindent}{3cm}
    \item \(q\in(0,3]\) in the one-dimensional case \((d=1)\);
    \item \(q\in[1,3]\) in the multidimensional case \((d\ge2)\).
\end{itemize}
The upper endpoint of the admissible index range, namely $q=3$, is sharp in every dimension.
\end{theorem}

The proof is given in Section \ref{sec:3}. The sharp upper endpoint $q=3$
follows from the sharp dimension-free constant $C_u=3$ in the functional
inequality \eqref{eq:74}. While the concavity of Tsallis entropy for
$q\in(0,3]$ was proved in dimension one in
\cite{hung2022generalization,wu2025completely}, the proof below is different
and extends directly to higher dimensions in the range $1\le q\le 3$. The
lower endpoint $q=1$ in higher dimensions is discussed in Remark \ref{rmk:20}.

Through the proof of Theorem \ref{thm:ieee1} in Section \ref{sec:3}, we also obtain the generalized de Bruijn identity and monotonicity of the associated \(q\)-Fisher information.
\begin{corollary}
Under the index range specified in Theorem \ref{thm:ieee1}, the following properties hold along the heat flow:
\begin{itemize}\addtolength{\itemindent}{3cm}
    \item \(\displaystyle \frac{d}{dt}S_q(\phi_t)=\frac{4}{q}\int \|\nabla \phi_t^{q/2}\|^2 dx=:I_q(\phi_t);\)
    \item \(\displaystyle \frac{d}{dt}I_q(\phi_t)\le0.\)
\end{itemize}
\end{corollary}
Here \(I_q\), up to normalization, is the Tsallis--Fisher or \(q\)-Fisher information.  The first identity is the generalized de Bruijn identity, while the second follows from the concavity of \(S_q(\phi_t)\).

\paragraph{Concavity of Tsallis Entropy Power}
The next results concern the entropy-power functional associated with Tsallis entropy.
\begin{definition}[Tsallis Entropy Power]\label{deff:ieee1}
For \(q\in(0,1)\cup(1,\infty)\) and \(\mu>0\), define
\begin{equation}
N_{q,\mu}(\rho)=\exp\left(\frac{\mu}{d}S_q(\rho)\right).
\end{equation}
At \(q=1\), we use the Shannon limit
\begin{equation}
N_{1,\mu}(\rho)=\exp\left(\frac{\mu}{d}H(\rho)\right).
\end{equation}
For \(q=1, \mu=2\), this is the classical Shannon entropy power.
\end{definition}

Costa's entropy-power concavity theorem \cite{costa1985new} states that
{\begin{equation}
\frac{d^2}{dt^2}N_{1,2}(\phi_t)\le0.
\end{equation}}
The entropy-power theorem below contains this result as the endpoint \(q=1\), \(\mu=2\).  We also recall that Savar\'e and Toscani~\cite{savare2014concavity} proved an entropy-power concavity theorem for R\'enyi entropy along the porous medium equation; the present result concerns Tsallis entropy along the linear heat flow.

\begin{theorem}\label{thmm:2}
Along the heat flow \(\phi_t\), the following statements hold.
\begin{itemize}\addtolength{\itemindent}{0cm}
    \item If \(q=1\) and \(0<\mu\le2\), then \(\displaystyle \frac{d^2}{dt^2}N_{1,\mu}(\phi_t)\le0\).  In particular, \(N_{1,2}\) is concave along the heat flow.
    \item If \(1<q\le2\) and
    \begin{equation}
    \int \phi_0^qdx\le \frac{d(q-1)+2(2-q)}{\mu},
    \end{equation}
    then \(\displaystyle \frac{d^2}{dt^2}N_{q,\mu}(\phi_t)\le0\) for all \(t>0\).
    \item If \(2\le q<3\) and
    \begin{equation}
    \int \phi_0^qdx\le \frac{d(3-q)}{\mu},
    \end{equation}
    then \(\displaystyle \frac{d^2}{dt^2}N_{q,\mu}(\phi_t)\le0\) for all \(t>0\).
\end{itemize}
\end{theorem}
The smallness condition on $\int \phi_0^q\,dx$ is natural in the study of
Tsallis entropy power along the heat flow. The heat-kernel example below shows
that, when $q>1$, concavity may fail at small times without such an assumption.
We do not claim that the threshold values obtained here are sharp. In the
limiting case $\mu=2$ and $q\to 1$, our result recovers Costa's EPI.

\begin{theorem}\label{thm:asymptotic-power}
Let \(\phi_0\) be an arbitrary probability measure and let \(\phi_t=\phi_0*p_t\).  For every \(1<q<3\) and \(\mu>0\), there exists an explicit constant \(T_{d,q,\mu}<\infty\), given in Section \ref{sec:5}, such that
\begin{equation}
\frac{d^2}{dt^2}N_{q,\mu}(\phi_t)\le0
\qquad\text{for all }t\ge T_{d,q,\mu}.
\end{equation}
\end{theorem}

The following example illustrates why the restriction on the initial condition and the large-time regime are necessary for the concavity of the Tsallis entropy power.

\begin{example}
	Along the heat kernel $p_t$, we have
	\begin{equation}
		\frac{d^2}{dt^2} N_{q,\mu}(p_t)
		=\lambda N_{q,\mu}(p_t)t^{-\left(\frac{d(q-1)}{2} + 2\right)}
		\left[ \lambda t^{-\frac{d(q-1)}{2}} - \left( \frac{d(q-1)}{2} + 1 \right) \right],
		\quad 
		\lambda=\frac{\mu}{2} q^{-d/2} (4\pi)^{-\frac{d(q-1)}{2}}.
	\end{equation}
	It follows that $\frac{d^2}{dt^2} N_{q,\mu}(p_t) \leq 0$ only for $t$ larger than a certain threshold, which can be explicitly determined from the above expression. For sufficiently small $t$, we have $\frac{d^2}{dt^2} N_{q,\mu}(p_t) > 0$, since in this regime $\int p_t^q dx$ does not satisfy the smallness condition imposed on the initial data in Theorem \ref{thmm:2}.
\end{example}

\paragraph{{Nonlinear Transform}}To prove the theorems, we introduce a transformation of \(\phi _{t}\) that allows for a unified calculation and estimation of the temporal derivatives of the Tsallis entropy under the heat equation.


In the following sections, without loss of generality, we let \(u_{t}=(u_{\delta })_{t}:=\phi_t^{q/2},q=\frac{2}{1+\delta}\),  as the parameter dependence on \(\delta \) is fixed and implicit throughout the analysis. We have the following lemma.
 \begin{lemma}
     We have $u_t$ satisfies the following non-linear equation
     \begin{equation}\label{eq:ieee2}
         \partial_tu_t=\Delta u_t+\delta\frac{\normsq{\nabla u_t}}{u_t}.
     \end{equation}
 \end{lemma}
 {
\begin{remark}
We note that although the transformation introduces a nonlinearity into
Equation \eqref{eq:ieee2}, it remains advantageous because the equation retains
a linear dependence on the parameter \(\delta\).  More importantly, the
transformation puts the second-derivative calculation into a much simpler
normal form.

Indeed, the transformation is not logically necessary in the sense that, for
every \(t>0\), the heat flow satisfies \(\phi_t>0\), and therefore
\[
u_t=\phi_t^{1/(1+\delta)}
\]
is in one-to-one correspondence with \(\phi_t\).  Thus any functional
inequality proved for \(u_t\) can be rewritten as an equivalent inequality in
the original variable \(\phi_t\), and conversely.  However, the pulled-back
inequality in \(\phi_t\) is considerably less transparent.  For example, if
one works directly with $\phi_t$,
then a direct differentiation along the heat equation gives
\[
\frac{d^2}{dt^2}S_q(\phi_t)
=
-q\left[
2\int \phi_t^{q-2}(\Delta\phi_t)^2dx
+
(q-2)\int \phi_t^{q-3}\|\nabla\phi_t\|^2\Delta\phi_t\,dx
\right].
\]
Thus the desired sign would require a \(q\)-dependent weighted inequality
involving the two terms above.  If the mixed term is further integrated by
parts, one obtains additional weighted expressions involving
\(\nabla^2\phi_t(\nabla\phi_t,\nabla\phi_t)\) and
\(\|\nabla\phi_t\|^4\), with weights depending on powers of \(\phi_t\).  These
terms obscure the underlying square structure and make the dependence on \(q\)
less transparent.

By contrast, after setting
\[
u_t=\phi_t^{1/(1+\delta)}=\phi_t^{q/2},
\qquad
q=\frac{2}{1+\delta},
\]
we have the simple identity
\[
\int \phi_t^q dx=\int u_t^2 dx.
\]
Moreover, \(u_t\) solves
\[
\partial_t u_t
=
\Delta u_t
+
\delta\frac{\|\nabla u_t\|^2}{u_t},
\]
where the parameter \(\delta\) enters linearly.  Consequently, the second
derivative of the entropy is reduced to estimating only one nonlinear term:
\[
\frac{d^2}{dt^2}S_q(\phi_t)
=
-4(1+\delta)
\left[
\int(\Delta u_t)^2dx
+
\delta\int \Delta u_t
\frac{\|\nabla u_t\|^2}{u_t}dx
\right].
\]
In other words, the sign problem becomes the study of
\[
D+\delta A,
\qquad
D:=\int(\Delta u_t)^2dx,
\qquad
A:=\int \Delta u_t\frac{\|\nabla u_t\|^2}{u_t}dx.
\]
The functional inequalities needed later are then the dimension-uniform
estimates
\[
D+A\ge0,
\qquad
A\le3D.
\]
The first identity has the exact square form
\[
D+A
=
\int u_t^2\|\nabla^2\log u_t\|^2dx,
\]
and the second estimate follows from an explicit integration-by-parts
sum-of-squares identity.  Hence the transformation separates the role of the
entropy parameter from the analytic inequalities: the parameter \(q\), or
equivalently \(\delta\), appears only as a scalar coefficient in \(D+\delta A\),
while the inequalities for \(D\) and \(A\) themselves are independent of \(q\).
This is the main technical benefit of introducing the nonlinear variable.  It
does not create a fundamentally different problem, but it reveals the simple
unweighted structure that is hidden in the original heat-equation variables. See Section \ref{sec:3} for detailed derivation.
\end{remark}
}

 \begin{proof}
     The proof follows from direct computation. We have 
     \begin{equation}\label{eq:ieee3}
    \partial_t u_t = p\phi_t^{p-1} \partial_t \phi_t = p\phi_t^{p-1} \Delta \phi_t,
\end{equation}the gradient of $u_t$ is
\begin{equation}
    \nabla u_t = p\phi_t^{p-1} \nabla \phi_t \implies \nabla \phi_t = \frac{1}{p} \phi_t^{1-p} \nabla u_t,
\end{equation}
the Laplacian of $u_t$ is
\begin{align}\label{eq:ieee5}
    \Delta u_t &= \nabla \cdot (p\phi_t^{p-1} \nabla \phi_t)= p(p-1)\phi_t^{p-2} \normsq{\nabla \phi_t} + p\phi_t^{p-1} \Delta \phi_t.
\end{align}
Substitute the time derivative from Equation \eqref{eq:ieee3} into Equation \eqref{eq:ieee5}, we have
\begin{equation}
    \Delta u_t = p(p-1)\phi_t^{p-2} \normsq{\nabla \phi_t} + \partial_t u_t.
\end{equation}
Substitute $\normsq{\nabla \phi_t} = \frac{1}{p^2} \phi_t^{2-2p} \normsq{\nabla u_t}$ into the equation, we get
\begin{align}
    \partial_t u_t &= \Delta u_t - p(p-1)\phi_t^{p-2} \left( \frac{1}{p^2} \phi_t^{2-2p} \normsq{\nabla u_t} \right) \\
    &= \Delta u_t - \frac{p-1}{p} \phi_t^{-p} \normsq{\nabla u_t}
\end{align}
Since $\phi_t^{-p} = \frac{1}{u_t}$ and $p = \frac{1}{1+\delta}$, the coefficient is
\begin{equation}
    -\frac{p-1}{p} = \frac{1-p}{p} = \frac{1}{p} - 1 = (1+\delta) - 1 = \delta,
\end{equation}
therefore, $u_t$ satisfies
\begin{equation}
    \partial_t u_t = \Delta u_t + \delta \frac{\normsq{\nabla u_t}}{u_t}.
\end{equation}
 \end{proof}

\section{Verification of Integration by Parts}\label{sec:2}
Further analysis of Equation \eqref{eq:ieee2} relies on the application of integration by parts. To ensure the validity of these operations—specifically the vanishing of boundary terms at infinity—we establish a series of preliminary estimates. \revised{Throughout the remainder of this paper, we assume that $\delta \in (-1,1]$ when $d>1$, while in the one-dimensional case $d=1$, we allow $\delta \in (-1,\infty)$.} The rigorous justification for these integration-by-parts identities is formalized in Proposition \ref{prop:11}, which serves as the primary objective of this subsection. To establish this result, we first derive several necessary auxiliary bounds.
\begin{lemma}\label{lem:9}
    For any \(t>0\) and \(p>1\), let \(q\) satisfy the conjugate exponent relation \(p^{-1}+q^{-1}=1\). The following pointwise estimate holds for the derivatives of the density \(\phi _{t}\):
    \begin{equation}\label{eq:ieee11}
    \max_{i,j,k}\{\frac{|\partial_i\phi_t|}{\phi_t^{\frac{1}{p}}}(x),\frac{|\partial_{ij}\phi_t|}{\phi_t^{\frac{1}{p}}}(x),\frac{|\partial_{ijk}\phi_t|}{\phi_t^{\frac{1}{p}}}(x)\}\leq \Big(\int Q^q_t(x-y)p_t(x-y)\phi_0dy\Big)^{\frac{1}{q}}\revised{\leq C^{\frac{1}{q}}_t\phi_{2t}^{\frac{1}{q}}},
\end{equation}
where \(p_{t}(x)\) denotes the heat kernel and \(Q_{t}(x)\) is a polynomial in \(x\) with time-dependent coefficients. 
\end{lemma}
\begin{proof}
Since 
\begin{equation}
    \phi_t=\int p_t(x-y)\phi_0(y)dy,
\end{equation}
so by direct computation, we have
\begin{equation}
    \begin{aligned}
        \partial_i\phi_t&=\int Q_i(x-y,t^{-1})p_t(x-y)\phi_0(y)dy\\
        \partial_{ij}\phi_t&=\int Q_{ij}(x-y,t^{-1})p_t(x-y)\phi_0(y)dy\\
        \partial_{ijk}\phi_t&=\int Q_{ijk}(x-y,t^{-1})p_t(x-y)\phi_0(y)dy,
    \end{aligned}
\end{equation}
for some polynomials $Q_i,Q_{ij},Q_{ijk}$. We will denote 
\begin{equation}
    Q_t(x-y)=\sum_{i,j,k}|Q_i(x-y,t^{-1})|+|Q_{ij}(x-y,t^{-1})|+|Q_{ijk}(x-y,t^{-1})|,
\end{equation}
then
\begin{equation}
    \begin{aligned}
        \max_{i,j,k}\{|\partial_i\phi_t|,|\partial_{ij}\phi_t|,|\partial_{ijk}\phi_t|\}
        &\leq \int Q_t(x-y)p_t(x-y)\phi_0dy\\
        &\leq \Big(\int Q^q_t(x-y)p_t(x-y)\phi_0dy\Big)^{\frac{1}{q}}\Big(\int p_t(x-y)\phi_0dy\Big)^{\frac{1}{p}}\\
        &=\Big(\int Q^q_t(x-y)p_t(x-y)\phi_0dy\Big)^{\frac{1}{q}}\phi_t^{\frac{1}{p}},
    \end{aligned}
\end{equation}
where in the second inequality we used the Cauchy-Schwarz inequality.
Thus  we have
\begin{equation}
    \max_{i,j,k}\{\frac{|\partial_i\phi_t|}{\phi_t^{\frac{1}{p}}},\frac{|\partial_{ij}\phi_t|}{\phi_t^{\frac{1}{p}}},\frac{|\partial_{ijk}\phi_t|}{\phi_t^{\frac{1}{p}}}\}\leq \Big(\int Q^q_t(x-y)p_t(x-y)\phi_0dy\Big)^{\frac{1}{q}},
\end{equation}
for any $t>0,p>1$ and $q$ satisfies $p^{-1}+q^{-1}=1$. \revised{The last inequality in \eqref{eq:ieee11} follows from the estimate $|Q(x)p_t(x)| \leq C_{t,\epsilon,Q} p_{t+\epsilon}(x)$, which holds for any polynomial $Q$ because the exponential decay of the heat kernel dominates any polynomial growth.Here $C_{t,\epsilon,Q} > 0$ is a constant depending on $t$, $\epsilon~(>0)$, and the polynomial $Q$.

}

\end{proof}

In order to simplify the upper bound in \eqref{eq:ieee11}, we introduce the following lemma, which establishes the necessary integral estimates.
\begin{lemma}\label{lem:10}
    We have
    \revised{\begin{equation}
        \int \phi_t^q(x)dx<\infty,
    \end{equation}
    for any $t\geq 0$ and $q\geq 1$.}
\end{lemma}
\begin{proof}
This is the classical \(L^q-L^1\) estimate for the heat
semigroup. Indeed, by Young's inequality,
\[
\|\phi_t\|_q
=
\|p_t*\phi_0\|_q
\le
\|p_t\|_q\|\phi_0\|_1.
\]
Therefore,
\begin{equation}
\|\phi_t\|_q^q
\le
\|p_t\|_q^q\|\phi_0\|_1^q
=
K_{d,q}t^{-\frac{d(q-1)}{2}},
\qquad
K_{d,q}:=q^{-d/2}(4\pi)^{-\frac{d(q-1)}{2}}.
\end{equation}
 
\end{proof}

By combining the results of the two preceding lemmas, we establish the following result:
\begin{lemma}\label{lem:ieee5}
    We have 
\begin{equation}\label{eq:123}
    \begin{aligned}
        \int |\partial_{i}{u_t}|^pdx<\infty,
        \int |\partial_{ij}{u_t}|^pdx<\infty,
        \int |\partial_{ijk}{u_t}|^pdx<\infty,
    \end{aligned}
\end{equation}
for any $p\geq 1+\delta,t>0$; and 
\begin{equation}\label{eq:121}
    \begin{aligned}
        \int |\partial_{i}\sqrt{u_t}|^pdx<\infty,
        \int |\partial_{ij}\sqrt{u_t}|^pdx<\infty,
        \int |\partial_{ijk}\sqrt{u_t}|^pdx<\infty,
    \end{aligned}
\end{equation}
for any $p\geq 2(1+\delta),t>0$.
\end{lemma}
\begin{proof}
Recalling the definition \(u_{t}:=\phi _{t}^{\frac{1}{1+\delta }}\), a direct computation yields:
\begin{equation}
    \begin{aligned}
        \partial_i u_t&=\frac{1}{1+\delta}\frac{\partial_i\phi_t}{\phi_t^{\frac{\delta}{1+\delta}}},\\
        \partial_{ij}{u_t}&=\frac{1}{1+\delta}\frac{\partial_{ij}\phi_t}{\phi^{\frac{\delta}{1+\delta}}_t}-\frac{\delta}{(1+\delta)^2}\frac{\partial_i\phi_t}{\phi_t^{\frac{2\delta+1}{2(1+\delta)}}}\frac{\partial_j\phi_t}{\phi_t^{\frac{2\delta+1}{2(1+\delta)}}},\\
        \partial_{ijk}{u_t}&=\frac{1}{1+\delta}\frac{\partial_{ijk}\phi_t}{\phi_t^{\frac{\delta}{1+\delta}}}-\frac{\delta}{(1+\delta)^2}\frac{\partial_{ij}\phi_t}{\phi_t^{\frac{2\delta+1}{2(1+\delta)}}}\frac{\partial_{k}\phi_t}{\phi_t^{\frac{2\delta+1}{2(1+\delta)}}}\\
        &\quad-\frac{\delta}{(1+\delta)^2}\frac{\partial_{ik}\phi_t}{\phi_t^{\frac{1+2\delta}{2(1+\delta)}}}\frac{\partial_{j}\phi_t}{\phi_t^{\frac{1+2\delta}{2(1+\delta)}}}-\frac{\delta}{(1+\delta)^2}\frac{\partial_{jk}\phi_t}{\phi_t^{\frac{1+2\delta}{2(1+\delta)}}}\frac{\partial_{i}\phi_t}{\phi_t^{\frac{1+2\delta}{2(1+\delta)}}}\\
        &\quad+\frac{\delta(1+2\delta)}{(1+\delta)^3}\frac{\partial_{i}\phi_t}{\phi_t^{\frac{2+3\delta}{3(1+\delta)}}}\frac{\partial_{j}\phi_t}{\phi_t^{\frac{2+3\delta}{3(1+\delta)}}}\frac{\partial_{k}\phi_t}{\phi_t^{\frac{2+3\delta}{3(1+\delta)}}},
    \end{aligned}
\end{equation}
thus with Lemma \ref{lem:9}, we have
\begin{equation}\label{eq:117}
    \begin{aligned}
        |\partial_i u_t|
        &\revised{\leq C(\delta,t)\phi_{2t}^{\frac{1}{1+\delta}}}\\
        |\partial_{ij}{u_t}|&\revised{\leq C(\delta,t)\phi_{2t}^{\frac{1}
        {1+\delta}}}\\
        |\partial_{ijk}u_t|
        &\revised{\leq C(\delta,t)\phi_{2t}^{\frac{1}{1+\delta}}},
    \end{aligned}
\end{equation}
with Lemma \ref{lem:10}, we further have
\begin{equation}
    \begin{aligned}
        \int |\partial_{i}{u_t}|^pdx<\infty,
        \int |\partial_{ij}{u_t}|^pdx<\infty,
        \int |\partial_{ijk}{u_t}|^pdx<\infty,
    \end{aligned}
\end{equation}
for any $p\geq 1+\delta$ and $t>0$.

The second part is similar.
Since $\sqrt{u_t}:=\phi_t^{\frac{1}{2(1+\delta)}}$, a direct computation yields:
\begin{equation}
    \begin{aligned}
        \partial_i\sqrt{u_t}&=\frac{1}{2(1+\delta)}\frac{\partial_i\phi_t}{\phi_t^{\frac{1+2\delta}{2(1+\delta)}}},\\
        \partial_{ij}\sqrt{u_t}&=\frac{1}{2(1+\delta)}\frac{\partial_{ij}\phi_t}{\phi^{\frac{1+2\delta}{2(1+\delta)}}_t}-\frac{1+2\delta}{4(1+\delta)^2}\frac{\partial_i\phi_t}{\phi_t^{\frac{3+4\delta}{4(1+\delta)}}}\frac{\partial_j\phi_t}{\phi_t^{\frac{3+4\delta}{4(1+\delta)}}},\\
        \partial_{ijk}\sqrt{u_t}&=\frac{1}{2(1+\delta)}\frac{\partial_{ijk}\phi_t}{\phi_t^{\frac{1+2\delta}{2(1+\delta)}}}-\frac{1+2\delta}{4(1+\delta)^2}\frac{\partial_{ij}\phi_t}{\phi_t^{\frac{3+4\delta}{4(1+\delta)}}}\frac{\partial_{k}\phi_t}{\phi_t^{\frac{3+4\delta}{4(1+\delta)}}}\\
        &\quad-\frac{1+2\delta}{4(1+\delta)^2}\frac{\partial_{ik}\phi_t}{\phi_t^{\frac{3+4\delta}{4(1+\delta)}}}\frac{\partial_{j}\phi_t}{\phi_t^{\frac{3+4\delta}{4(1+\delta)}}}-\frac{1+2\delta}{4(1+\delta)^2}\frac{\partial_{jk}\phi_t}{\phi_t^{\frac{3+4\delta}{4(1+\delta)}}}\frac{\partial_{i}\phi_t}{\phi_t^{\frac{3+4\delta}{4(1+\delta)}}}\\
        &\quad+\frac{(1+2\delta)(3+4\delta)}{8(1+\delta)^3}\frac{\partial_{i}\phi_t}{\phi_t^{\frac{5+6\delta}{6(1+\delta)}}}\frac{\partial_{j}\phi_t}{\phi_t^{\frac{5+6\delta}{6(1+\delta)}}}\frac{\partial_{k}\phi_t}{\phi_t^{\frac{5+6\delta}{6(1+\delta)}}},
    \end{aligned}
\end{equation}
thus with Lemma \ref{lem:9}, we have
\begin{equation}\label{eqq:117}
    \begin{aligned}
        |\partial_{i}\sqrt{u_t}|
        &\revised{\leq C(\delta,t)\phi_{2t}^{\frac{1}
        {2(1+\delta)}}},\\
        |\partial_{ij}\sqrt{u_t}|
        &\revised{\leq C(\delta,t)\phi_{2t}^{\frac{1}
        {2(1+\delta)}}}\\
        |\partial_{ijk}\sqrt{u_t}|
        &\revised{\leq C(\delta,t)\phi_{2t}^{\frac{1}
        {2(1+\delta)}}},
    \end{aligned}
\end{equation}
with Lemma \ref{lem:10}, we further have
\begin{equation}
    \begin{aligned}
        \int |\partial_{i}\sqrt{u_t}|^pdx<\infty,
        \int |\partial_{ij}\sqrt{u_t}|^pdx<\infty,
        \int |\partial_{ijk}\sqrt{u_t}|^pdx<\infty,
    \end{aligned}
\end{equation}
for any $p\geq 2(1+\delta)$ and $t>0$.
\end{proof}

Equipped with these preliminary bounds, we are positioned to verify the vanishing of boundary integrals, thereby ensuring the validity of the integration-by-parts operations. We have the following proposition.
\begin{proposition}\label{prop:11}
    \revised{For $d=1$, we take $\delta \in (-1,\infty)$, while for $d \ge 2$, we take $\delta \in (-1,1]$. Under these assumptions,} we have
    \begin{equation}
        \int u_t\Delta u_tdx=-\int\normsq{\nabla u_t}dx,\int \Delta \normsq{\nabla u_t}dx=0,\int \inner{\nabla\Delta u_t}{\nabla u_t}dx=-\int\normsq{\Delta u_t}dx,
    \end{equation}
    and
    \begin{equation}
        \int \Delta u_t\normsq{\nabla\sqrt{u_t}}dx=-\int \inner{\nabla u_t}{\nabla\normsq{\nabla\sqrt{u_t}}}dx=\int u_t\Delta\normsq{\nabla \sqrt{u_t}}dx.
    \end{equation}
\end{proposition}
\begin{proof}
\revised{We first consider $\delta\in (-1,1],d\geq 1$.}   Applying integration by parts over the cube \(Q_{R}(0):=[-R,R]^{d}\), we must demonstrate that the resulting boundary integrals vanish in the limit \(R\rightarrow \infty \). Specifically, we are required to verify that:
    \begin{equation}\label{eqq:31}\footnotesize
        \begin{aligned}
         &\lim_{R\to\infty}\int_{\partial Q_R(0)}\norm{\frac{\partial u_t}{\partial n}}u_tdS\leq\lim_{R\to\infty}\Big[\int_{\partial Q_R(0)}\normsq{\nabla u_t}dS+\int_{\partial Q_R(0)}u_t^2dS\Big]\to 0,\\
         &\lim_{R\to\infty}\int_{\partial Q_R(0)}\norm{\frac{\partial \normsq{\nabla u_t}}{\partial n}}dS\leq\lim_{R\to\infty}\Big[\int_{\partial Q_R(0)}\normsq{\nabla u_t}dS+\int_{\partial Q_R(0)}\normsq{\nabla^2u_t}dS\Big]\to 0,\\
          &\lim_{R\to\infty}\int_{\partial Q_R(0)}\norm{\frac{\partial u_t}{\partial n}\Delta u_t}dS\leq\lim_{R\to\infty}\Big[\int_{\partial Q_R(0)}\normsq{\nabla u_t}dS+\int_{\partial Q_R(0)}\normsq{\nabla^2u_t}dS\Big]\to 0,\\
            &\lim_{R\to\infty}\int_{\partial Q_R(0)}\norm{\frac{\partial u_t}{\partial n}}\normsq{\nabla\sqrt{u_t}}dS\leq\lim_{R\to\infty}\Big[\int_{\partial Q_R(0)}\normsq{\nabla u_t}dS+\int_{\partial Q_R(0)}\norm{\nabla\sqrt{u_t}}^4dS\Big]\to 0,\\
            &\lim_{R\to\infty}\int_{\partial Q_R(0)}u_t\norm{\frac{\partial \normsq{\nabla\sqrt{u_t}}}{\partial n}}dS  \leq \lim_{R\to\infty}\Big[\int_{\partial Q_R(0)}u_t^2dS +\int_{\partial Q_R(0)}\norm{\nabla^2\sqrt{u_t}}^4dS+\int_{\partial Q_R(0)}\norm{\nabla\sqrt{u_t}}^4dS   \Big]  \to 0,\\
        \end{aligned}
    \end{equation}
    where we have employed the Cauchy-Schwarz inequality, and \(n\) denotes the unit outward normal vector to the boundary \(\partial Q_{R}(0)\).
    
    By the trace theorem~(see \cite[Section 5.5]{evans2022partial}), there exists a constant \(C>0\), independent of both \(y\) and \(f\), such that for any \(f\in W^{1,p}(Q_{1/2}(y))\) with \(p\ge 1\), we have
    \begin{equation}\label{eq:ieee34}
    	\int_{\partial Q_{\frac{1}{2}}(y)}\norm{f}^pdS\leq C\int_{Q_{\frac{1}{2}}(y)}\norm{f}^p+\norm{\nabla f}^pdx,
    \end{equation}
    here $Q_{\frac{1}{2}}(y)=y+[-1/2,1/2]^d$. Consider a tiling of the annular region \(Q_{R}(0)\setminus Q_{R-1}(0)\) by a finite collection of essentially disjoint cubes \(\{Q_{1/2}(y_{i})\}_{i}\), such that \(Q_{R}(0)\setminus Q_{R-1}(0)=\bigcup _{i}Q_{1/2}(y_{i})\) and \(\text{int}(Q_{1/2}(y_{i}))\cap \text{int}(Q_{1/2}(y_{j}))=\emptyset \) for \(i\ne j\). Utilizing this partition and the trace inequality~\eqref{eq:ieee34}, we obtain:
    \begin{equation}\label{eq:128}\footnotesize
    	\begin{aligned}
    	\int_{\partial Q_R(0)}{u_t^2}+\normsq{\nabla u_t}+\normsq{\nabla^2u_t}dS&\leq \sum_i \int_{\partial Q_{\frac{1}{2}}(y_i)}{u_t^2}+\normsq{\nabla u_t}+\normsq{\nabla^2u_t}dS\\
    	&\leq 3C\int_{Q_R(0)\setminus Q_{R-1}(0)} u_t^2+\normsq{\nabla u_t}+\normsq{\nabla^2 u_t}+\normsq{\nabla^3u_t}dx\\
    	\int_{\partial Q_R(0)}\norm{\nabla \sqrt{u_t}}^4+\norm{\nabla^2 \sqrt{u_t}}^4dS&\leq \sum_i \int_{\partial Q_{\frac{1}{2}}(y_i)}\norm{\nabla \sqrt{u_t}}^4+\norm{\nabla^2 \sqrt{u_t}}^4dS\\
    	&\leq 2C\int_{Q_R(0)\setminus Q_{R-1}(0)}\norm{\nabla \sqrt{u_t}}^4+\norm{\nabla^2 \sqrt{u_t}}^4+\norm{\nabla^3 \sqrt{u_t}}^4dx.
    	\end{aligned}
    \end{equation}
    Due to \eqref{eq:123}, \eqref{eq:121}, $\delta\leq 1$ and the $L^p-L^1$ estimates for the heat equation, that is $\norm{\phi_t}_{L^p}\leq Ct^{-\frac{d}{2}(1-\frac{1}{p})}\norm{\phi_0}_{L^1}$, for $p\geq 1$, we have
    \begin{equation}
    	\begin{aligned}
    		&\int u_t^2+\normsq{\nabla u_t}+\normsq{\nabla^2 u_t}+\normsq{\nabla^3u_t}dx<\infty,\\
    		&\int\norm{\nabla \sqrt{u_t}}^4+\norm{\nabla^2 \sqrt{u_t}}^4+\norm{\nabla^3 \sqrt{u_t}}^4dx<\infty,
    	\end{aligned}
    \end{equation}
    thus the right hand sides of \eqref{eq:128} vanish as $R\to\infty$, and we finished the proof.

    \revised{When $d=1$, there is a more direct argument to show that the
boundary term on the right-hand side of \eqref{eqq:31} vanishes 
as $R \to \infty$, for any $\delta\in (-1,\infty)$. In one dimension, the boundary integral 
reduces to the sum of the integrand evaluated at the endpoints 
$x=R$ and $x=-R$. By \eqref{eq:117} and \eqref{eqq:117}, it 
therefore suffices to verify that $\phi_t(x) \to 0$ as $|x| \to \infty$, for any $t \ge 0$.

Fix $\varepsilon > 0$. Since $\phi_0 \in L^1(\mathbb{R})$, there exists $R>0$ such that
\[
\int_{|y|>R} |\phi_0(y)|\,dy < \varepsilon.
\]
Split
\[
\phi_t(x) = \int_{|y|\le R} p_t(x-y)\phi_0(y)\,dy + \int_{|y|>R} p_t(x-y)\phi_0(y)\,dy
=: I_1(x) + I_2(x).
\]

For the tail term,
\[
|I_2(x)| \le \|p_t\|_\infty \int_{|y|>R} |\phi_0(y)|\,dy \le C_t \varepsilon.
\]

For the compact part, if $|y|\le R$, then $|x-y|\ge |x|-R$, hence
\[
p_t(x-y) \le \frac{1}{\sqrt{4\pi t}} e^{-\frac{(|x|-R)^2}{4t}}.
\]
Thus,
\[
|I_1(x)| \le \left(\int_{|y|\le R} |\phi_0(y)|\,dy\right) \frac{1}{\sqrt{4\pi t}} e^{-\frac{(|x|-R)^2}{4t}} \to 0
\quad \text{as } |x|\to\infty.
\]

Combining the estimates,
\[
\limsup_{|x|\to\infty} |\phi_t(x)| \le C_t \varepsilon.
\]
Since $\varepsilon>0$ is arbitrary, we conclude
\[
\phi_t(x) \to 0 \quad \text{as } |x|\to\infty.
\]}
\end{proof}

It is important to recognize that the integrability condition $\int |f|dx<\infty $ does not, by itself, guarantee that the boundary integral \(\int _{\partial Q_{R}(0)}|f|dS\) vanishes in the limit \(R\rightarrow \infty \). To properly justify the neglect of these boundary terms, one may instead invoke the trace theorem within appropriate weighted Sobolev spaces.

\revised{

When $d \ge 2$, although $\phi_t(x) \to 0$ as $\|x\| \to \infty$, the boundary integral
\[
\int_{\partial B_R(0)} \phi_t^p \, dS
\]
does not necessarily vanish as $R \to \infty$ for $p \in (0,1)$.
For simplicity, we use $\partial B_R$ instead of $Q_R$, here $B_R$ is the ball of
radius $R$ centered at the origin. Consequently, 
in dimensions $d \ge 2$, the method used in this work only allows 
verification of integration by parts for $\delta \in (-1,1]$, but 
fails for $\delta > 1$~(which corresponds to $q\in (0,1)$).
A counterexample is
\[
\phi_0(x) = \frac{C}{(1+\|x\|)^{d+1}} \in L^1(\mathbb{R}^d).
\]
Then, for $\|x\| = R$, we have
\begin{equation}
\begin{aligned}
\phi_t(x) &= \int_{\|y-x\| \le 1} p_t(x-y) \phi_0(y) \, dy 
           + \int_{\|y-x\| > 1} p_t(x-y) \phi_0(y) \, dy \\
&\ge \int_{\|y-x\| \le 1} p_t(x-y) \phi_0(y) \, dy \\
&\gtrsim \frac{1}{R^{d+1}}.
\end{aligned}
\end{equation}
Hence, the boundary integral satisfies
\[
\int_{\partial B_R} \phi_t^p \, dS \gtrsim \int_{\partial B_R} R^{-p(d+1)} \, dS \gtrsim R^{d-1 - p(d+1)} \to \infty, \quad \text{if } p < \frac{d-1}{d+1}.
\]

}

\section{Estimation of the Second-Order Temporal Derivative of \(\int u_t^2dx\)}\label{sec:3}
Throughout this section we write \(u=u_t\) when no confusion can arise and set
\begin{equation}\label{eq:def-DABTG}
D:=\int(\Delta u)^2dx,\qquad
A:=\int \Delta u\frac{\|\nabla u\|^2}{u}dx,\qquad
B:=\int\frac{\|\nabla u\|^4}{u^2}dx,
\end{equation}
\begin{equation}
T:=\int\frac{\nabla^2u(\nabla u,\nabla u)}{u}dx,
\qquad
G:=\int\|\nabla u\|^2dx.
\end{equation}
All integrations are over \(\mathbb{R}^d\).  The integration-by-parts identities below are justified by Proposition \ref{prop:11}; in particular, for \(d\ge2\) we use them in the range \(\delta\in(-1,1]\), while in dimension one they are justified for \(\delta\in(-1,\infty)\).

Recall that
\[
q=\frac{2}{1+\delta},\qquad u_t=\phi_t^{1/(1+\delta)}=\phi_t^{q/2},
\]
and that Lemma 8 gives
\begin{equation}\label{eq:nonlinear-ut-section4}
\partial_tu_t=\Delta u_t+\delta\frac{\|\nabla u_t\|^2}{u_t}.
\end{equation}
We first compute the first derivative:
\begin{equation}\label{eq:first-U-derivative}
\begin{aligned}
\frac{d}{dt}\frac12\int u_t^2dx
&=\int u_t\partial_tu_tdx \\
&=\int u_t\Delta u_tdx+\delta\int\|\nabla u_t\|^2dx \\
&=-(1-\delta)\int\|\nabla u_t\|^2dx.
\end{aligned}
\end{equation}
Since \(\int u_t^2dx=\int\phi_t^qdx\) and \(q-1=(1-\delta)/(1+\delta)\), this is equivalent to the generalized de Bruijn identity
\begin{equation}\label{eq:debruijn-section4}
\frac{d}{dt}S_q(\phi_t)=2(1+\delta)\int\|\nabla u_t\|^2dx
=\frac{4}{q}\int\|\nabla\phi_t^{q/2}\|^2dx.
\end{equation}
For the second derivative, using \eqref{eq:nonlinear-ut-section4} again,
\begin{equation}\label{eq:51}
\begin{aligned}
\frac{d}{dt}\frac12\int\|\nabla u_t\|^2dx
&=\int\langle\nabla u_t,\nabla\partial_tu_t\rangle dx \\
&=-\int\Delta u_t\left(\Delta u_t+\delta\frac{\|\nabla u_t\|^2}{u_t}\right)dx \\
&=-D-\delta A.
\end{aligned}
\end{equation}
Therefore
\begin{equation}\label{eq:S-second-D-A}
\frac{d^2}{dt^2}S_q(\phi_t)=-4(1+\delta)(D+\delta A).
\end{equation}
Thus the concavity of \(S_q(\phi_t)\) is reduced to proving
\begin{equation}\label{eq:concavity-reduction}
D+\delta A\ge0.
\end{equation}
Equivalently, if one has constants \(C_l,C_u\ge0\) such that
\begin{equation}\label{eq:74}
-C_l\int(\Delta u)^2dx
\le
\int\Delta u\frac{\|\nabla u\|^2}{u}dx
\le
C_u\int(\Delta u)^2dx,
\end{equation}
then \(D+\delta A\ge0\) for \(\delta\in[-1/C_u,1/C_l]\).  The main point of the present section is that the upper constant can be taken to be the sharp dimension-free value
\begin{equation}\label{eq:sharp-cu-three}
C_u=3.
\end{equation}

\begin{proposition}[lower bound]\label{lem:hessian-identities}
Under the assumptions of Proposition \ref{prop:11},
\begin{equation}\label{eq:hessian-laplace}
\int\|\nabla^2u\|^2dx=\int(\Delta u)^2dx=D,
\end{equation}
\begin{equation}\label{eq:A-B-T}
A=B-2T,
\end{equation}
and
\begin{equation}\label{eq:lower-square-identity}
D+A=
\int\left\|\nabla^2u-\frac{\nabla u\otimes\nabla u}{u}\right\|^2dx
=
\int u^2\|\nabla^2\log u\|^2dx\ge0.
\end{equation}
Consequently,
\begin{equation}\label{eq:lower-Cl-one}
A\ge -D.
\end{equation}
\end{proposition}
\begin{remark}
    Identity \eqref{eq:lower-square-identity} is essentially a reformulation of an identity already appearing in \cite{villani2006short}.
\end{remark}
\begin{proof}
Equality \eqref{eq:hessian-laplace} follows from Proposition \ref{prop:11}. More precisely, by the Bochner formula
\[
\frac{1}{2}\Delta\normsq{\nabla u}=
\normsq{\nabla^2 u}
+
\inner{\nabla\Delta u}{\nabla u},
\]
together with the integration-by-parts identities established in Proposition \ref{prop:11}, we obtain the desired identity.

Next,
\begin{equation}
\begin{aligned}
A
&=\int\Delta u\frac{\|\nabla u\|^2}{u}dx \\
&=-\int\left\langle\nabla u,\nabla\left(\frac{\|\nabla u\|^2}{u}\right)\right\rangle dx \\
&=-2\int\frac{\nabla^2u(\nabla u,\nabla u)}{u}dx
+\int\frac{\|\nabla u\|^4}{u^2}dx \\
&=B-2T.
\end{aligned}
\end{equation}
Finally, using \eqref{eq:hessian-laplace} and \eqref{eq:A-B-T},
\begin{equation}
\begin{aligned}
D+A
&=\int\|\nabla^2u\|^2dx
-2\int\frac{\nabla^2u(\nabla u,\nabla u)}{u}dx
+\int\frac{\|\nabla u\|^4}{u^2}dx \\
&=\int\left\|\nabla^2u-\frac{\nabla u\otimes\nabla u}{u}\right\|^2dx.
\end{aligned}
\end{equation}
Since
\[
\nabla^2u-\frac{\nabla u\otimes\nabla u}{u}=u\nabla^2\log u,
\]
we obtain \eqref{eq:lower-square-identity}.  The lower bound \eqref{eq:lower-Cl-one} follows immediately.
\end{proof}

\begin{proposition}[Sharp upper bound]\label{lem:sharp-upper-bound}
Under the assumptions of Proposition \ref{prop:11},
\begin{equation}\label{eq:A-upper-three-D}
\int\Delta u\frac{\|\nabla u\|^2}{u}dx
\le
3\int(\Delta u)^2dx.
\end{equation}
More precisely, the estimate follows from the exact identity
\begin{equation}\label{eq:Cu3-exact-SOS}
\begin{aligned}
3D-A
&=\int\left[
2\left(\Delta u-\frac{\|\nabla u\|^2}{3u}\right)^2
+
\left\|\nabla^2u-\frac{\nabla u\otimes\nabla u}{3u}\right\|^2
\right]dx.
\end{aligned}
\end{equation}
The constant \(3\) is sharp in every dimension.
\end{proposition}

\begin{proof}
We prove the upper bound by an explicit square-completion identity.  Define the symmetric matrices and their traces
\begin{equation}\label{eq:PQab-definition}
P:=\nabla^2u,
\qquad
Q:=\frac{\nabla u\otimes\nabla u}{u},
\qquad
a:=\operatorname{tr}P=\Delta u,
\qquad
b:=\operatorname{tr}Q=\frac{\|\nabla u\|^2}{u}.
\end{equation}
Since \(Q\) is rank one, \(\|Q\|^2=b^2\).  Moreover
\begin{equation}\label{eq:PQ-T-density}
P:Q=\frac{\nabla^2u(\nabla u,\nabla u)}{u},
\end{equation}
where \(:\) denotes the Frobenius inner product.  We claim the following pointwise algebraic identity:
\begin{equation}\label{eq:pointwise-SOS-Cu3}
\begin{aligned}
&2\left(a-\frac b3\right)^2+
\left\|P-\frac Q3\right\|^2 \\
&\quad =
3a^2-ab+
\left(\|P\|^2-a^2\right)
+\frac13\left(b^2-ab-2P:Q\right).
\end{aligned}
\end{equation}
Indeed, expanding the left-hand side gives
\[
2a^2-\frac43ab+\frac29b^2
+
\|P\|^2-\frac23P:Q+\frac19\|Q\|^2,
\]
and using \(\|Q\|^2=b^2\) gives the right-hand side of \eqref{eq:pointwise-SOS-Cu3}.

We now integrate \eqref{eq:pointwise-SOS-Cu3}.  Equality \eqref{eq:hessian-laplace} gives
\begin{equation}\label{eq:null-lagrangian-one}
\int\left(\|P\|^2-a^2\right)dx
=
\int\|\nabla^2u\|^2dx-
\int(\Delta u)^2dx
=0.
\end{equation}
Equality \eqref{eq:A-B-T} gives
\begin{equation}\label{eq:null-lagrangian-two}
\int\left(b^2-ab-2P:Q\right)dx
=
B-A-2T
=0.
\end{equation}
  Therefore, after integration, \eqref{eq:pointwise-SOS-Cu3} becomes exactly
\begin{equation}\label{eq:SOS-integrated-Cu3}
3D-A
=\int\left[
2\left(\Delta u-\frac{\|\nabla u\|^2}{3u}\right)^2
+
\left\|\nabla^2u-\frac{\nabla u\otimes\nabla u}{3u}\right\|^2
\right]dx.
\end{equation}
The right-hand side is a sum of squares, hence nonnegative.  This proves
\[
A\le3D,
\]
which is \eqref{eq:A-upper-three-D}.  
In dimension one, \eqref{eq:A-B-T} also gives the useful identity
\begin{equation}\label{eq:one-dimensional-3A-B}
3A=B=\int\frac{(u')^4}{u^2}dx\ge0,
\end{equation}
because in one dimension \(T=A\).  This identity will be used below in the one-dimensional concavity argument.

We now prove that the constant \(3\) is sharp.  Recall the notation
\[
D(u):=\int (\Delta u)^2\,dx,
\qquad
A(u):=\int \Delta u\,\frac{\|\nabla u\|^2}{u}\,dx .
\]
The inequality proved above is
\[
A(u)\le 3D(u).
\]
To show that the constant \(3\) cannot be improved, it suffices to construct
smooth positive admissible functions \(u_\varepsilon\) such that
\[
\frac{A(u_\varepsilon)}{D(u_\varepsilon)}\longrightarrow 3.
\]

We first identify the local extremizing profile.  The identity shows that
equality in \(A(u)= 3D(u)\) would force the square terms to vanish.  Thus,
at least formally and locally, one expects
\begin{equation}\label{eq:equality-equation-Cu3}
\nabla^2u=\frac{\nabla u\otimes\nabla u}{3u},
\qquad
\Delta u=\frac{\|\nabla u\|^2}{3u}.
\end{equation}
In one dimension this reduces to
\begin{equation}\label{eq:one-dimensional-extremal-ode}
u''=\frac{(u')^2}{3u}.
\end{equation}
Let \(w=u^{2/3}\).  Then
\[
w''
=
\frac{2}{3}u^{-1/3}u''
-
\frac{2}{9}u^{-4/3}(u')^2.
\]
Using \eqref{eq:one-dimensional-extremal-ode}, we obtain
\[
w''
=
\frac{2}{3}u^{-1/3}\frac{(u')^2}{3u}
-
\frac{2}{9}u^{-4/3}(u')^2
=0.
\]
Hence \(w\) is affine, and therefore the local positive solutions are of the
form
\[
u(x)=c(x+a)^{3/2}
\]
on intervals where \(x+a>0\), with \(c>0\).  For such profiles,
\[
\frac{(u')^2}{u}=3u'',
\]
and consequently
\begin{equation}\label{eq:pointwise-extremal-ratio}
u''\frac{(u')^2}{u}=3(u'')^2
\end{equation}
pointwise.  Thus the ratio \(A(u)/D(u)\) is locally equal to \(3\).

These local profiles are not globally admissible on \(\mathbb R\).  We therefore
regularize the singularity and attach the profile to a rapidly decaying
background.  Let \(\chi\in C_c^\infty((-2,2))\) satisfy
\(0\le \chi\le 1\) and \(\chi\equiv 1\) on \((-1,1)\).  Define
\[
r_\varepsilon(x):=(x^2+\varepsilon^2)^{3/4}
\]
and
\begin{equation}\label{eq:sharpness-test-function}
u_\varepsilon(x)
=
e^{-x^2}
+
\chi(x)\left(r_\varepsilon(x)-e^{-x^2}\right).
\end{equation}
Then \(u_\varepsilon>0\), \(u_\varepsilon\in C^\infty(\mathbb R)\), and
\(u_\varepsilon=e^{-x^2}\) outside \((-2,2)\).  Hence
\(u_\varepsilon\in L^p(\mathbb R)\) for every \(p>0\), and all boundary terms
vanish in the integrations by parts.

On \(|x|<1\), we have \(u_\varepsilon=r_\varepsilon\).  Set
\[
s=x^2+\varepsilon^2.
\]
A direct calculation gives
\[
r_\varepsilon'(x)
=
\frac{3}{2}x\,s^{-1/4},
\]
and
\[
r_\varepsilon''(x)
=
\frac{3}{4}(x^2+2\varepsilon^2)s^{-5/4}.
\]
Therefore
\[
(r_\varepsilon'')^2
=
\frac{9}{16}
(x^2+2\varepsilon^2)^2
(x^2+\varepsilon^2)^{-5/2},
\]
while
\[
r_\varepsilon''
\frac{(r_\varepsilon')^2}{r_\varepsilon}
=
\frac{27}{16}
x^2(x^2+2\varepsilon^2)
(x^2+\varepsilon^2)^{-5/2}.
\]

On the logarithmic region \(2\varepsilon\le |x|\le 1\), we have
\[
(r_\varepsilon'')^2
=
\frac{9}{16}\frac{1}{|x|}
+
O\!\left(\frac{\varepsilon^2}{|x|^3}\right),
\]
and
\[
r_\varepsilon''
\frac{(r_\varepsilon')^2}{r_\varepsilon}
=
\frac{27}{16}\frac{1}{|x|}
+
O\!\left(\frac{\varepsilon^2}{|x|^3}\right).
\]
Since
\[
\int_{2\varepsilon}^1 \frac{dx}{x}
=
\log(1/\varepsilon)+O(1),
\]
and
\[
\int_{2\varepsilon}^1 \frac{\varepsilon^2}{x^3}\,dx
=
O(1),
\]
we obtain
\[
\int_{2\varepsilon\le |x|\le 1}
(r_\varepsilon'')^2\,dx
=
\frac{9}{8}\log(1/\varepsilon)+O(1),
\]
and
\[
\int_{2\varepsilon\le |x|\le 1}
r_\varepsilon''
\frac{(r_\varepsilon')^2}{r_\varepsilon}\,dx
=
\frac{27}{8}\log(1/\varepsilon)+O(1).
\]

The remaining regions contribute only \(O(1)\).  Indeed, on
\(|x|\le 2\varepsilon\), the change of variables \(x=\varepsilon z\) gives
\[
r_\varepsilon(x)=\varepsilon^{3/2}(z^2+1)^{3/4},
\qquad
r_\varepsilon'(x)=\varepsilon^{1/2}R_1(z),
\qquad
r_\varepsilon''(x)=\varepsilon^{-1/2}R_2(z),
\]
where \(R_1\) and \(R_2\) are smooth and bounded for \(|z|\le 2\).  Hence
\[
\int_{|x|\le 2\varepsilon}(r_\varepsilon'')^2\,dx=O(1),
\]
and
\[
\int_{|x|\le 2\varepsilon}
r_\varepsilon''
\frac{(r_\varepsilon')^2}{r_\varepsilon}\,dx
=O(1).
\]
On the cutoff region \(1\le |x|\le 2\), the functions
\(u_\varepsilon\), \(u_\varepsilon'\), and \(u_\varepsilon''\) are uniformly
bounded in \(\varepsilon\), and \(u_\varepsilon\) is uniformly bounded away
from zero.  Therefore the cutoff contribution is \(O(1)\).  Outside
\((-2,2)\), we have \(u_\varepsilon=e^{-x^2}\), so the contribution is finite
and independent of \(\varepsilon\).

Consequently,
\[
D(u_\varepsilon)
=
\int_{\mathbb R}(u_\varepsilon'')^2\,dx
=
\frac{9}{8}\log(1/\varepsilon)+O(1),
\]
and
\[
A(u_\varepsilon)
=
\int_{\mathbb R}
u_\varepsilon''
\frac{(u_\varepsilon')^2}{u_\varepsilon}\,dx
=
\frac{27}{8}\log(1/\varepsilon)+O(1).
\]
Therefore
\[
\frac{A(u_\varepsilon)}{D(u_\varepsilon)}
=
\frac{
\frac{27}{8}\log(1/\varepsilon)+O(1)
}{
\frac{9}{8}\log(1/\varepsilon)+O(1)
}
\longrightarrow 3.
\]
Thus no constant smaller than \(3\) can hold in dimension one.

For \(d\ge2\), we embed the same one-dimensional sequence into higher
dimensions.  Write \(x=(x_1,y)\in\mathbb R\times\mathbb R^{d-1}\), and let
\(m=d-1\).  Choose a positive rapidly decaying function
\(h\in C^\infty(\mathbb R^m)\), for example
\[
h(y)=e^{-\|y\|^2/2}.
\]
For \(L>0\), set
\[
h_L(y):=h(y/L),
\qquad
U_{\varepsilon,L}(x_1,y):=u_\varepsilon(x_1)h_L(y).
\]
Then \(U_{\varepsilon,L}\) is smooth, positive, rapidly decaying, and belongs
to \(L^p(\mathbb R^d)\) for every \(p>0\).

The leading terms in \(D(U_{\varepsilon,L})\) and \(A(U_{\varepsilon,L})\)
come from the \(x_1\)-derivatives.  Indeed,
\[
\Delta U_{\varepsilon,L}
=
u_\varepsilon''h_L
+
u_\varepsilon\Delta h_L,
\]
and
\[
\frac{\|\nabla U_{\varepsilon,L}\|^2}{U_{\varepsilon,L}}
=
\frac{(u_\varepsilon')^2}{u_\varepsilon}h_L
+
u_\varepsilon\frac{\|\nabla h_L\|^2}{h_L}.
\]
Using the scaling \(h_L(y)=h(y/L)\), we have
\[
\int_{\mathbb R^m}h_L^2\,dy
=
L^m\int_{\mathbb R^m}h^2\,dy,
\]
whereas every integral involving at least one derivative of \(h_L\) is lower
order in \(L\).  More precisely,
\[
\int_{\mathbb R^m}\|\nabla h_L\|^2\,dy=O(L^{m-2}),
\qquad
\int_{\mathbb R^m}h_L\Delta h_L\,dy=O(L^{m-2}),
\]
and
\[
\int_{\mathbb R^m}(\Delta h_L)^2\,dy=O(L^{m-4}),
\qquad
\int_{\mathbb R^m}
\Delta h_L\frac{\|\nabla h_L\|^2}{h_L}\,dy
=O(L^{m-4}).
\]
Therefore, for fixed \(\varepsilon>0\),
\[
D(U_{\varepsilon,L})
=
L^m\left(\int_{\mathbb R^m}h^2\,dy\right)D(u_\varepsilon)
+
O_\varepsilon(L^{m-2}),
\]
and
\[
A(U_{\varepsilon,L})
=
L^m\left(\int_{\mathbb R^m}h^2\,dy\right)A(u_\varepsilon)
+
O_\varepsilon(L^{m-2}).
\]
Dividing these two estimates and sending \(L\to\infty\), we obtain
\[
\lim_{L\to\infty}
\frac{A(U_{\varepsilon,L})}{D(U_{\varepsilon,L})}
=
\frac{A(u_\varepsilon)}{D(u_\varepsilon)}.
\]
Finally, sending \(\varepsilon\to0\), we get
\[
\lim_{\varepsilon\to0}\lim_{L\to\infty}
\frac{A(U_{\varepsilon,L})}{D(U_{\varepsilon,L})}
=
3.
\]
Hence the same sharp constant \(C_u=3\) is forced in every dimension.
\end{proof}

A byproduct of the preceding proposition is the following sharp functional
inequality, which may be of independent analytic interest.

\begin{lemma}
Under the assumptions of Proposition \ref{prop:11}, one has
\begin{equation}\label{eq:grad-four-upper-nine-D}
\int \frac{\|\nabla u\|^4}{u^2}\,dx
\le
9\int (\Delta u)^2\,dx .
\end{equation}
The constant \(9\) is sharp in every dimension.
\end{lemma}

\begin{proof}
We have
\[
A=B-2T,
\]
or equivalently
\[
B=A+2T.
\]
Therefore, by the Cauchy--Schwarz inequality,
\[
\begin{aligned}
B
&=A+2T \\
&\le |A|+2|T| \\
&\le
\left(\int (\Delta u)^2\,dx\right)^{1/2}
\left(\int \frac{\|\nabla u\|^4}{u^2}\,dx\right)^{1/2} \\
&\quad
+2
\left(\int \|\nabla^2 u\|^2\,dx\right)^{1/2}
\left(\int \frac{\|\nabla u\|^4}{u^2}\,dx\right)^{1/2}.
\end{aligned}
\]
Moreover, we have
\[
\int \|\nabla^2 u\|^2\,dx
=
\int (\Delta u)^2\,dx.
\]
Hence
\[
B\le 3D^{1/2}B^{1/2}.
\]
If \(B=0\), the desired inequality is immediate. Otherwise, dividing by
\(B^{1/2}\) gives
\[
B^{1/2}\le 3D^{1/2},
\]
and therefore
\[
B\le 9D.
\]

It remains to prove sharpness. By the Cauchy--Schwarz inequality,
\[
A^2
=
\left(
\int \Delta u\frac{\|\nabla u\|^2}{u}\,dx
\right)^2
\le
\left(\int(\Delta u)^2\,dx\right)
\left(\int\frac{\|\nabla u\|^4}{u^2}\,dx\right)
=
DB.
\]
Suppose that the constant \(9\) in \eqref{eq:grad-four-upper-nine-D} could be
replaced by some \(C<9\). Then
\[
B\le CD.
\]
Combining this with \(A^2\le DB\), we would obtain
\[
A^2\le CD^2,
\]
and hence
\[
A\le \sqrt{C}\,D.
\]
Since \(\sqrt{C}<3\), this contradicts the sharpness of the constant \(3\) in
the inequality
\[
A\le 3D.
\]
Therefore no constant smaller than \(9\) can hold. Hence the constant \(9\) is
sharp in every dimension.
\end{proof}

We can now prove Theorem \ref{thm:ieee1}.

\begin{proof}[Proof of Theorem \ref{thm:ieee1}]
By \eqref{eq:S-second-D-A}, it is enough to prove \(D+\delta A\ge0\).

First consider \(d\ge2\).  Proposition \ref{prop:11} justifies the identities in the range \(\delta\in(-1,1]\).  If \(0\le\delta\le1\), then \eqref{eq:lower-square-identity} gives
\begin{equation}\label{eq:delta-positive}
D+\delta A=(1-\delta)D+\delta(D+A)\ge0.
\end{equation}
If \(-1/3\le\delta\le0\), then the sharp upper bound \eqref{eq:A-upper-three-D} gives
\begin{equation}\label{eq:delta-negative}
D+\delta A\ge D+3\delta D=(1+3\delta)D\ge0.
\end{equation}
Therefore \(S_q(\phi_t)\) is concave for
\[
-\frac13\le\delta\le1.
\]
Since \(q=2/(1+\delta)\), this is exactly \(1\le q\le3\).

In dimension one, Proposition \ref{prop:11} is available for all \(\delta>-1\).  Moreover, \eqref{eq:one-dimensional-3A-B} gives \(A=B/3\ge0\).  Hence \(D+\delta A\ge0\) for every \(\delta\ge0\).  For \(-1/3\le\delta\le0\), the same argument as \eqref{eq:delta-negative} applies.  Thus the admissible interval is \(\delta\in[-1/3,\infty)\), which is equivalent to \(q\in(0,3]\).  This proves the theorem.
\end{proof}
\begin{remark}
In the one-dimensional case ($d=1$), \cite{hung2022generalization,wu2025completely} studied the signs of the time derivatives of Tsallis entropy (and R\'enyi entropy in \cite{wu2025completely}) along the heat flow up to fourth order for certain ranges of the parameter $q$. Notably, they found that the Tsallis entropy is concave along the heat flow for $q \in (0,3]$, consistent with the results presented here. For higher-order time derivatives and the corresponding ranges of $q$ in one dimension, interested readers are referred to \cite{hung2022generalization,wu2025completely}.
\end{remark}
\begin{remark}[Sharpness of the constant and the endpoint \(q=3\)]
The constant \(C_u=3\) in \eqref{eq:74} is sharp in every dimension by Proposition \ref{lem:sharp-upper-bound}.  Consequently, the lower endpoint \(\delta=-1/3\), equivalently the upper endpoint \(q=3\), cannot be improved.  Indeed, suppose that for some \(\delta<-1/3\) one had
\[
D+\delta A\ge 0
\]
for all admissible \(u\). Since \(\delta<0\), this is equivalent to
\[
A\le -\frac1\delta D.
\]
But \(\delta<-1/3\) implies \(-1/\delta<3\). Hence there exists
\(\varepsilon\in(0,3)\) such that
\[
A\le (3-\varepsilon)D.
\]
This contradicts the sharpness of the constant \(3\) in the functional inequality
\[
A\le CD.
\]
\end{remark}

\begin{remark}[On the range \(q<1\) in higher dimensions]\label{rmk:20}
For \(d\ge2\), the proof above is restricted to \(q\ge1\) because the integration-by-parts justification in Proposition \ref{prop:11} is available for \(\delta\le1\).  This is not only a technical inconvenience.  Along the heat kernel \(p_t\),
\[
S_q(p_t)=\frac{1-q^{-d/2}(4\pi t)^{\frac d2(1-q)}}{q-1},
\]
and therefore
\[
\frac{d^2}{dt^2}S_q(p_t)
=\frac{d}{2t^2}\left(\frac{d(1-q)}2-1\right)
q^{-d/2}(4\pi t)^{\frac d2(1-q)}.
\]
Thus, whenever \(0<q<1-2/d\), the second derivative is positive along the heat kernel.  In particular, for \(d\ge3\), Tsallis entropy cannot be concave along the heat flow for the entire range \(q<1\).
\end{remark}

\section{Concavity of Tsallis Entropy Power Along Heat Flow}\label{sec:5}
This section proves Theorems \ref{thmm:2} and \ref{thm:asymptotic-power}.  We keep the notation
\[
U:=\int u_t^2dx=\int\phi_t^qdx,
\qquad
G:=\int\|\nabla u_t\|^2dx,
\]
and use \(D,A\) as in \eqref{eq:def-DABTG}.  We also set
\begin{equation}\label{eq:E-def}
E:=D+A=
\int u_t^2\|\nabla^2\log u_t\|^2dx.
\end{equation}
From \eqref{eq:debruijn-section4} and \eqref{eq:S-second-D-A},
\begin{equation}\label{eqq:81}
\begin{aligned}
\frac{d^2}{dt^2}N_{q,\mu}(\phi_t)
&=\left[\frac{\mu}{d}\frac{d^2}{dt^2}S_q(\phi_t)
+\left(\frac{\mu}{d}\frac{d}{dt}S_q(\phi_t)\right)^2\right]N_{q,\mu}(\phi_t) \\
&=\frac{4\mu(1+\delta)}{d}N_{q,\mu}(\phi_t)
\left[-D-\delta A+\frac{\mu(1+\delta)}{d}G^2\right].
\end{aligned}
\end{equation}
The prefactor is positive, so the sign is determined by the bracket.

We shall use two elementary estimates for \(G^2\).  First, by integration by parts and Cauchy--Schwarz,
\begin{equation}\label{eq:G2-UD}
G^2=\left(-\int u_t\Delta u_tdx\right)^2
\le \left(\int u_t^2dx\right)\left(\int(\Delta u_t)^2dx\right)
=UD.
\end{equation}
Second, using \(E=\int u_t^2\|\nabla^2\log u_t\|^2dx\),
\begin{equation}\label{eq:G2-UE}
G^2\le\frac d4 UE.
\end{equation}
Indeed,
\[
\int u_t^2\Delta\log u_tdx
=\int u_t\Delta u_tdx-\int\|\nabla u_t\|^2dx=-2G,
\]
and hence, by Cauchy--Schwarz and \((\operatorname{tr}M)^2\le d\|M\|^2\),
\[
4G^2
\le U\int u_t^2(\Delta\log u_t)^2dx
\le dU\int u_t^2\|\nabla^2\log u_t\|^2dx=dUE.
\]

\paragraph{The Shannon endpoint.}
At \(q=1\), we have \(\delta=1\), \(u_t=\sqrt{\phi_t}\), \(U=\int\phi_tdx=1\), and \(S_1=H\).  The bracket in \eqref{eqq:81} becomes
\[
-E+\frac{2\mu}{d}G^2.
\]
Using \eqref{eq:G2-UE} with \(U=1\),
\begin{equation}\label{eq:shannon-power-bracket}
-E+\frac{2\mu}{d}G^2
\le
-E+\frac{2\mu}{d}\frac d4E
=-\left(1-\frac\mu2\right)E.
\end{equation}
Therefore \(N_{1,\mu}\) is concave for \(0<\mu\le2\).  In particular, \(\mu=2\) gives Costa's entropy-power concavity.

\paragraph{The range \(1<q\le2\).}
Here \(0\le\delta<1\).  Since \(E=D+A\),
\begin{equation}\label{eq:positive-delta-bracket}
-D-\delta A=-(1-\delta)D-\delta E.
\end{equation}
We combine \eqref{eq:G2-UD} and \eqref{eq:G2-UE}.  For any \(\theta\in[0,1]\),
\begin{equation}\label{eq:convex-G}
G^2\le U\left(\theta D+(1-\theta)\frac d4E\right).
\end{equation}
Choose
\begin{equation}\label{eq:theta-choice}
\theta=\frac{d(1-\delta)}{d(1-\delta)+4\delta}.
\end{equation}
Then \eqref{eq:convex-G} gives
\begin{equation}\label{eq:positive-delta-positive-term}
\frac{\mu(1+\delta)}{d}G^2
\le
\frac{\mu(1+\delta)U}{d(1-\delta)+4\delta}\big((1-\delta)D+\delta E\big).
\end{equation}
Consequently the bracket in \eqref{eqq:81} is nonpositive whenever
\begin{equation}\label{eq:U-condition-delta-positive}
U\le\frac{d(1-\delta)+4\delta}{\mu(1+\delta)}.
\end{equation}
In terms of \(q=2/(1+\delta)\), this condition is
\begin{equation}\label{eq:U-condition-q-1-2}
\int\phi_t^qdx\le\frac{d(q-1)+2(2-q)}{\mu},
\qquad 1\le q\le2.
\end{equation}
At \(q=1\), this becomes \(1\le2/\mu\), exactly the Shannon endpoint above.

\paragraph{The range \(2\le q<3\).}
Here \(-1/3<\delta\le0\).  By the sharp estimate \eqref{eq:A-upper-three-D},
\begin{equation}\label{eq:negative-delta-entropy-term}
-D-\delta A\le-(1+3\delta)D.
\end{equation}
Combining this with \eqref{eq:G2-UD}, the bracket in \eqref{eqq:81} is bounded above by
\begin{equation}
\left[-(1+3\delta)+\frac{\mu(1+\delta)}{d}U\right]D.
\end{equation}
Thus the bracket is nonpositive whenever
\begin{equation}\label{eq:U-condition-delta-negative}
U\le\frac{d(1+3\delta)}{\mu(1+\delta)}.
\end{equation}
In terms of \(q\), this is
\begin{equation}\label{eq:U-condition-q-2-3}
\int\phi_t^qdx\le\frac{d(3-q)}{\mu},
\qquad 2\le q<3.
\end{equation}

\paragraph{Proof of Theorem \ref{thmm:2}.}
For \(q>1\), the quantity \(U_t=\int\phi_t^qdx\) is nonincreasing along the heat flow.  Indeed, \eqref{eq:first-U-derivative} gives
\begin{equation}\label{eq:U-monotone}
\frac{d}{dt}U_t=-2(1-\delta)G\le0,
\end{equation}
because \(q>1\) is equivalent to \(\delta<1\).  Therefore, if the initial datum satisfies the corresponding smallness condition in \eqref{eq:U-condition-q-1-2} or \eqref{eq:U-condition-q-2-3}, then the same condition holds for all \(t>0\), and \eqref{eqq:81} gives \(\frac{d^2}{dt^2}N_{q,\mu}(\phi_t)\le0\).  The case \(q=1\) was proved in \eqref{eq:shannon-power-bracket}.  This proves Theorem \ref{thmm:2}.

\paragraph{Asymptotic concavity.}
Let \(q>1\) and let \(\phi_0\) be any probability measure.  The heat-kernel \(L^q-L^1\) estimate gives
\begin{equation}\label{eq:heat-Lq-bound}
U_t=\|\phi_t\|_q^q\le\|p_t\|_q^q\|\phi_0\|_1^q
=K_{d,q}t^{-\frac{d(q-1)}2},
\qquad
K_{d,q}:=q^{-d/2}(4\pi)^{-\frac{d(q-1)}2}.
\end{equation}
Therefore, for \(1<q\le2\), the condition \eqref{eq:U-condition-q-1-2} holds whenever
\begin{equation}\label{eq:T-asymptotic-1-2}
t\ge
T_{d,q,\mu}^{(1)}
:=\left(\frac{\mu K_{d,q}}{d(q-1)+2(2-q)}\right)^{\frac{2}{d(q-1)}}.
\end{equation}
For \(2\le q<3\), the condition \eqref{eq:U-condition-q-2-3} holds whenever
\begin{equation}\label{eq:T-asymptotic-2-3}
t\ge
T_{d,q,\mu}^{(2)}
:=\left(\frac{\mu K_{d,q}}{d(3-q)}\right)^{\frac{2}{d(q-1)}}.
\end{equation}
Thus Theorem \ref{thm:asymptotic-power} holds with
\begin{equation}
T_{d,q,\mu}=
\begin{cases}
T_{d,q,\mu}^{(1)}, & 1<q\le2,\\[1mm]
T_{d,q,\mu}^{(2)}, & 2\le q<3.
\end{cases}
\end{equation}
At \(q=2\), the two formulas coincide.

\begin{remark}
The entropy-power result is not merely a consequence of the concavity of \(S_q(\phi_t)\).  The second derivative \eqref{eqq:81} contains the additional positive term \(\frac{\mu(1+\delta)}{d}G^2\).  The estimates \eqref{eq:G2-UD} and \eqref{eq:G2-UE}, especially the trace estimate leading to \eqref{eq:G2-UE}, are what allow the proof to include the Shannon entropy-power endpoint \(q=1\).
\end{remark}

\section{Conclusion}\label{sec:4}

In this work, we studied the evolution of Tsallis entropy along the heat flow and proved concavity in arbitrary dimensions for a sharp upper range of the entropic index. In dimension one, the admissible range recovers the known interval $q\in(0,3]$. In higher dimensions, we prove concavity for $q\in[1,3]$. The upper endpoint $q=3$ is sharp in every dimension. The proof is based on a nonlinear transformation of the heat equation, a sharp dimension-free functional inequality with constant $C_u=3$, and a rigorous verification of the integration-by-parts identities required in the argument.

The method also clarifies the relation between the present approach and SOS techniques. The key upper bound is obtained from an explicit integration-by-parts sum-of-squares identity. Thus the proof is compatible with the broad SOS philosophy, but differs from conventional computer-assisted SOS approaches: the square identity is constructed analytically through the transformed variable, rather than through a high-dimensional Gram-matrix or semidefinite-programming search. This provides a relatively simple route to a dimension-uniform estimate and explains why the endpoint $q=3$ appears naturally.

Several consequences follow from the main concavity theorem. We recover the generalized de Bruijn identity for Tsallis entropy and prove monotonicity of the associated $q$-Fisher information along the heat flow. We also establish concavity results for Tsallis entropy power. These results include the Shannon entropy-power case and recover Costa's EPI at the classical endpoint. For general initial data, we prove an asymptotic concavity result, showing that Tsallis entropy power becomes concave after sufficiently long heat-flow smoothing. The heat-kernel example shows that such restrictions are natural in the entropy-power problem, since concavity may fail at small times when $q>1$.

The results have potential relevance for information-theoretic and analytic problems involving Gaussian smoothing beyond the Shannon setting. Since heat flow corresponds to adding Gaussian noise, the monotonicity and concavity properties proved here provide tools for studying non-additive entropy under noise, especially for heavy-tailed or non-Gaussian distributions. They may also be useful in the study of generalized entropy power inequalities, Gaussian extremality, large-noise asymptotics, and generalized Fisher-information dissipation.

Several questions remain open. The lower endpoint in higher dimensions, especially the behavior of Tsallis entropy for $q<1$, is more delicate and cannot be treated by the present integration-by-parts argument in full generality. It would also be interesting to determine sharp entropy-power thresholds, to extend the method to other entropy functionals such as R\'enyi entropy, and to study other diffusion equations such as porous medium or Langevin-type flows. Finally, in view of the recent counterexamples to complete-monotonicity-type conjectures, future extensions to higher-order derivatives should be formulated as precise finite-order and parameter-dependent problems, rather than as unrestricted complete-monotonicity statements.

\section*{Conflict of Interest }
There is no conflict of interest.
\section*{Acknowledgment}
LS acknowledges support from the Munich Center for Machine Learning .


  
   
    

\ifCLASSOPTIONcaptionsoff
  \newpage
\fi

\bibliography{bib}
\bibliographystyle{IEEEtran}

\end{document}